\documentclass[jmp,amsmath,amssymb]{revtex4-1}
\usepackage{graphicx}
\usepackage{bm}
\usepackage{amsmath} 
\usepackage[utf8]{inputenc}
\usepackage[utf8]{inputenc}
\usepackage{tikz-cd}
\usepackage{graphicx}
\usepackage{subcaption}
\usepackage{amsmath} 
\usepackage{amsmath}
\usepackage[english]{babel}
\usepackage{marvosym}
\usepackage{amssymb}
\usepackage{graphicx}
\usepackage[section]{placeins}
\usepackage{hyperref}
\usepackage{mathtools}
\usepackage{fullpage}
\usepackage[utf8]{inputenc}

\usepackage{amssymb}
\usepackage{caption}
\usepackage{subcaption}
\usepackage[section]{placeins}
\usepackage{mathtools}
\usepackage[all]{xy}

\usepackage[english]{babel}
\usepackage{amsthm}
\theoremstyle{plain}
\usepackage{color, soul}
\newtheorem*{theorem*}{Theorem}
\newtheorem{theorem}{Theorem}[section]

\newtheorem{proposition}[theorem]{Proposition}

\newtheorem{remark}[theorem]{Remark}

\newtheorem{example}[theorem]{Example}

\newtheorem{definition}[theorem]{Definition}

\newtheorem{procedure}[theorem]{Procedure}

\usepackage{mathtools}

\usepackage{mathtools}
\usepackage{amsmath,amscd,amssymb,amsfonts,amsthm,hyperref}

\usepackage[all]{xy}

\DeclarePairedDelimiter\ket{\lvert}{\rangle}
\DeclarePairedDelimiterX\braket[2]{\langle}{\rangle}{#1 \delimsize\vert #2}

\usepackage[T1]{fontenc}
\usepackage{mathptmx}
\usepackage{etoolbox}

\makeatletter
\def\@email#1#2{%
 \endgroup
 \patchcmd{\titleblock@produce}
  {\frontmatter@RRAPformat}
  {\frontmatter@RRAPformat{\produce@RRAP{*#1\href{mailto:#2}{#2}}}\frontmatter@RRAPformat}
  {}{}
}%
\makeatother
\preprint{AIP/123-QED}

\begin{document}

\title{Phase spaces that cannot be cloned in classical mechanics}
\author{Yuan Yao}
    \email{yuan\_yao@berkeley.edu}

\affiliation{ 
Department of Mathematics. University of California, Berkeley. Berkeley 94720, CA.
}%

\date{\today}
\begin{abstract}
 The quantum no cloning theorem is an essential result in quantum information theory. Following this idea, we give a physically natural definition of cloning in the context of classical mechanics using symplectic geometry, building on work of Fenyes. We observe, following Fenyes, any system with phase space $(\mathbb{R}^{2N}, dx_i\wedge dy_i)$ can be cloned in our definition. However, we show that if $(M,\omega)$ can be cloned in our definition, then $M$ must be contractible. For instance, this shows the simple pendulum cannot be cloned in Hamiltonian mechanics. We further formulate a robust notion of approximate cloning, and show that if $(M, \omega)$ can be approximately cloned, then $M$ is contractible. We give interpretations of our results and in some special cases reconcile our no cloning theorems with the general experience that classical information is clonable. Finally we point to new directions of research, including a connection of our result with the classical measurement problem.
\end{abstract}
\maketitle
\section{Introduction} \label{section:intro}

The impossibility of cloning quantum states is well understood \cite{RevModPhys.77.1225,fenyes} and essential in the theory of quantum information. It is a question of great interest to understand whether cloning processes are possible in other physical systems; and there have been many such results. The impossibility of cloning for classical statistical systems is established in \cite{PhysRevLett.88.210601}. In \cite{Clone_nonlinear_Ham}, it is shown that cloning is possible for a certain class of nonlinear extensions of quantum mechanics they call ``nonlinear Hamiltonian quantum mechanics,'' but is not possible for another extension of quantum mechanics they call ``mean field hybrid Hamiltonian mechanics.''

Whether cloning is possible in the context of classical mechanics was raised in John Baez's 2008 classical mechanics class lectures (See \cite{fenyes_2010,baez_2008}). Further
study by Fenyes \cite{fenyes} gave a general formulation of cloning in classical physics allowing for the presence of a cloning machine; in their formulation it is shown that systems with phase space $(\mathbb{R}^{2N}, dx_i \wedge dy_i)$ can always be cloned. More concrete realizations of cloning in classical mechanics are given in follow up work in \cite{REDDY2019125846} where they produce an explicit Hamiltonian that can be used to clone states in $(\mathbb{R}^{2},dx\wedge dy)$ as well as experiments to realize classical cloning using nonlinear optics. 
For an overall summary of cloning in various quantum and classical contexts, see also the discussion in \cite{TEH201247}.

To remind the reader of the general setup of cloning problems, we first give a review of the quantum no-cloning theorem in section \ref{sec:quantum cloning}. Briefly speaking, cloning is the process where we start with a quantum state $\ket{b}$, let it interact with a quantum state $\ket{\psi}$ (the system to be cloned), and the cloning machine $\ket{r}$. The final output is $\ket{\psi}\otimes \ket{\psi} \otimes \ket{r'}$. The process transforms $\ket{b}$ into $\ket{\psi}$; and $\ket{r'}$ is the final state of the cloning machine.
Then in analogy with the quantum cloning process, we introduce our notion of classical cloning in Definition \ref{def:clone}. Our definition is slightly more restrictive than that given in Fenyes \cite[Definition 4]{fenyes} or the procedure used in \cite{REDDY2019125846} (see Remark \ref{remark:comparison} for a concrete example), but includes everything that is permitted in Hamiltonian dynamics. We shall explain in Section \ref{sec:classical_cloning} why our formulation is cloning is more natural as it excludes symplectic transformations that are not physical.

We observe $(\mathbb{R}^{2N}, dx_i\wedge dy_i)$ continues to admit cloning in our new definition in Proposition \ref{prop:R2n}. Next we prove our main theorem:
\begin{theorem*}[See Theorem \ref{thm:clone}] If $(M,\omega)$ can be cloned as in Definition \ref{def:clone}, then $M$ must be contractible.
\end{theorem*}
This means clonable phase spaces cannot have any topology. While our main tool is a homotopy group computation and Whitehead's theorem, we now give an illustration of our methods in the case of the simple pendulum. In this case, consider the red non-contractible curve on the two dimensional torus shown in Figure \ref{red}. In the context of our proof this red non-contractible curve essentially comes from a non-contractible curve in the phase of the pendulum $T^*S^1$. If a cloning were to exist, after chasing through the definition, it would imply a continuous deformation from the red curve in Figure \ref{red} to the blue curve in Figure \ref{blue}, which is impossible. In general a similar argument shows a clonable phase space cannot contain any non-contractible spheres of any dimension, and Whitehead's theorem states any space with this property is contractible.

\begin{figure}[h!]
     \centering
     \begin{subfigure}[b] {0.36\textwidth} 
         \centering
        \includegraphics[width=\textwidth]{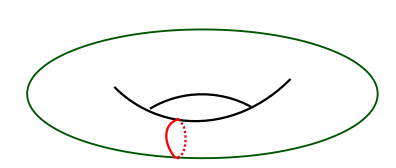}
         \caption{red curve on a torus}
         \label{red}
     \end{subfigure}
     \begin{subfigure}[b] {0.3\textwidth}
         \centering
         \includegraphics[width=1.3\textwidth]{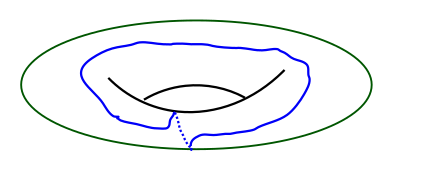}
         \caption{blue curve on a torus}
         \label{blue}
     \end{subfigure}
        \caption{The red curve cannot be deformed continuously to the blue curve.}
        \label{fig:three graphs}
\end{figure}

We next turn to the problem of approximate cloning. Even though it has been shown quantum cloning is not possible, there has been considerable interest in trying to understand whether approximate cloning is possible, see for instance \cite{optimal_quantum_cloning_machines,quantumcloning,optimal_universal_quantum_cloning,quantum_copying,quantum_cloning_machines,quantum_cloning_machine_for_equatorial_qubits,duan1997nonorthogonal,information_theoretic_limits,Probabilistic_Cloning_and_Identification_of_Linearly_Independent_Quantum_States}. The point is that if one where to develop a quantum cryptography protocol based on the no cloning theorem, one would like to understand how much information an eavesdropper can access despite the no-cloning theorem. Approximate quantum cloning protocols have been constructed with various degrees of fidelity \cite{duan1997nonorthogonal,optimal_quantum_cloning_machines,optimal_universal_quantum_cloning,Probabilistic_Cloning_and_Identification_of_Linearly_Independent_Quantum_States,quantum_copying,quantum_cloning_machine_for_equatorial_qubits,quantum_cloning_machines}.

We examine the analogous problem in classical mechanics. We give a formulation of approximate cloning in classical mechanics adapted to our setup in Definition \ref{def:approximate cloning}. Since our methods are inherently topological, we are able to show the following strong obstruction:

\begin{theorem*}[See Theorem \ref{thm:approximately clone}] If $(M,\omega)$ can be approximately cloned as in Definition \ref{def:approximate cloning}, then $M$ must be contractible.
\end{theorem*}

Our results about non-cloning are highly counter-intuitive. In the  world around us we clone classical states all the time: a computer can copy data without worrying about the topology of phase spaces. In section \ref{physical interpretation of our results}, we take the first step in re-conciliating our theorems with our intuition. The key observation is that in the real world we are very often only cloning discrete points in phase space, and almost never entire regions of phase space that may have topology. We observe the quantum analogue is that we can always clone one dimensional subspaces of a Hilbert space, which we explain in Section \ref{subsec:a quantum analogy}.

Finally we formulate some questions for future investigation. We give a summary here and leave the detailed formulations to section \ref{sec:future directions}.

\begin{enumerate}
\item There are exotic symplectic structures on $\mathbb{R}^{2n}$ that are not symplectomorphic to the standard symplectic structure, see for example \cite{exotic_symplectic}. Can cloning still happen in these exotic geometries?
\item There has been considerable interest in the classical thermodynamics of information processing \cite{information_thermo,landauer,thermodynamics_of_feedback_controlled_systems}. Classical cloning can also be thought of as a type of information processing. The question is then is there a right notion of ``free energy'' that can be associated to the cloning process?
\item Many aspects of our results still remain counter-intuitive. In an idealized classical physics experiment, one can ask an experimenter to clone a system (for example a pendulum) by first measuring its position and momentum without disturbing the system, then pushing an identical system (e.g. another pendulum) into the same position and momentum; however, this is a violation of our no-cloning theorem.

In section \ref{sec:future directions}, we point to several directions of research required to reconcile the above with our no-cloning theorem by further examining the classical process of measurement.
\end{enumerate}
\section{Quantum Cloning}\label{sec:quantum cloning}

There are many possible definitions of cloning. We take the definition that allows for the presence of a cloning machine\footnote{There is another definition that involves no cloning machines, i.e. given $\ket{\psi}, \ket{b} \in \mathcal{H}$, we consider (unitary) maps $\mathcal{H}\otimes \mathcal{H} \rightarrow \mathcal{H} \otimes \mathcal{H}$ taking $\ket{\psi}\otimes \ket{b} \rightarrow \ket{\psi} \otimes \ket{\psi}$. Such maps are not only prohibited in quantum mechanics, their classical analogues are also not possible. See \cite{fenyes_2010,fenyes} for a discussion.}. Let $\ket{\psi} \in \mathcal{H}$ be a state vector in a Hilbert space which we think of the system to be cloned, and let $\ket{b} \in \mathcal{H}$ be another state vector. Let $\ket{r} \in \mathcal{H}'$ be a state vector which we think of the cloning machine. Then the cloning process is an unitary map $\mathcal{H} \otimes \mathcal {H} \otimes \mathcal{H}' \rightarrow \mathcal{H} \otimes \mathcal {H} \otimes \mathcal{H}'$ taking
\[
\ket{\psi} \otimes \ket{b} \otimes \ket{r} \rightarrow \ket{\psi} \otimes \ket{\psi} \otimes \ket{r'}.
\]
Here we think of $\ket{r'}\in \mathcal{H}'$ as the state of the cloning machine after the cloning process. 

The quantum no cloning theorem then says
\begin{proposition}[proposition 3 in \cite{fenyes}, Section I A \cite{RevModPhys.77.1225}]
Let $\mathcal{H}$ and $\mathcal{H}'$ both be finite dimensional complex Hilbert spaces, with $\dim \mathcal{H} >1$. There cannot exist vectors $\ket{b} \in \mathcal{H}, \ket{r} \in \mathcal{H}'$ and an unitary map $U: \mathcal{H} \otimes \mathcal{H} \otimes \mathcal{H}'\rightarrow \mathcal{H} \otimes \mathcal{H} \otimes \mathcal{H}'$ taking
\[
\ket{\psi} \otimes \ket{b} \otimes \ket{r} \rightarrow \ket{\psi} \otimes \ket{\psi} \otimes \ket{r'}
\]
for every $\ket{\psi} \in \mathcal{H}$. Here $\ket{r'} \in \mathcal{H}'$ depends on $\ket{\psi}$.
\end{proposition}
Heuristically this is saying we cannot build a machine to perform quantum cloning on two systems. 

\section{Classical Cloning} \label{sec:classical_cloning}
We give a classical formulation of the above cloning process. In classical mechanics we replace the Hilbert space with the phase space (i.e. symplectic manifolds) and unitary maps with symplectomorphisms. Our definition is a modification of the classical cloning process given in \cite[Definition 4]{fenyes}. Our definition is slightly more restrictive but more natural from a classical mechanics point of view. 

\begin{definition} \label{def:clone}
A classical cloning process is given by
\begin{itemize}
    \item Symplectic manifolds $(M, \omega_1)$ (the phase space of systems to be cloned) and $(N, \omega_2)$ (the phase space of the cloning machine). 
    \item A point $b \in M$, a point $r \in N$, and a symplectomorphism $\phi$ in the identity component of $\textup{Sympl}(M\times M \times N)$ that sends
    \[
    (x,b,r) \rightarrow (x,x,f(x,b,r))
    \]
    for all $x \in M$. Here $f(-,b,r)$ is some smooth function from $M$ to $N$.
\end{itemize}
If the above data exists for the symplectic manifold $(M,\omega_1)$, then we say it is \textbf{clonable}.
\end{definition}
The only distinction between our definition and that of Definition 4 in \cite{fenyes} is that we require $\phi$ to lie in the identity component of $\textup{Sympl}(M\times M \times N)$, which means $\phi$ is connected by a one parameter family of symplectomorphisms to the identity map. This is a very natural condition to impose, because if we want our cloning process to evolve in a continuous manner, the resulting $\phi$ must lie in the identity component. To be specific, in classical mechanics, time evolution of the systems should be described by the flow of a (potentially time dependent) Hamiltonian, and such flows must lie in the identity component. It is not clear what the non-identity components of $\textup{Sympl}(M\times M \times N)$ should correspond to physically.

We first observe the cloning process introduced in in \cite[Example 1,2]{fenyes} continues to work in this new definition.
\begin{proposition}[See Example 2 in \cite{fenyes}] \label{prop:R2n}
    If $(M, \omega_1) \cong (\mathbb{R}^{2N},\omega_0)$, where $\omega_0$ is the canonical symplectic form on $\mathbb{R}^{2N}$, then $(M,\omega_1)$ can be cloned in the sense of Definition \ref{def:clone}.
\end{proposition}

\begin{remark}
In \cite{REDDY2019125846} they provide an explicit (time dependent) Hamiltonian in the case of $(\mathbb{R}^2,dx_i \wedge dy_i)$. Their computations can be extended to show $(\mathbb{R}^{2N},dx_i\wedge dy_i)$ can be cloned using time dependent Hamiltonians.
\end{remark}

However, we also exhibit topological constraints to systems that can be cloned in the sense of Definition \ref{def:clone}.

\begin{theorem}\label{thm:clone}
If $(M,\omega)$ is clonable, then $M$ is contractible.
\end{theorem}

\begin{proof}
    Let $k\geq 1$ be an integer, consider the homotopy groups
    \[
    \pi_k (M\times M \times N) \cong \pi_k(M) \times  \pi_k(M) \times \pi_k(N).
    \]
For ease of notation we don't mention the base point. 

Suppose $\pi_k(M)$ is nontrivial, then consider $g:S^k \rightarrow M$ representing a nonzero homotopy class $[g]\in \pi_k(M)$. Consider
the map $(g,b,r): S^k\times \{pt\} \times \{pt\} \rightarrow M \times M \times N$ sending \[(x,b,r) \rightarrow (g(x), b,r).\] This gives a nonzero element of $\pi_k(M\times M \times N)$ of the form \[([g],0,0) \in \pi_k(M)\times \pi_k(M) \times \pi_k(N).\]

Next consider the composition of $(g,b,r)$ with $\phi \in \textup{Sympl}(M\times M \times N)$; for $(x,b,r) \in S^k \times \{pt\} \times \{pt\}$ this is given by
\[
\phi \circ (g,b,r): (x,b,r) \rightarrow (g(x),g(x), r').
\]
This also represents an element in $\pi_k(M\times M \times N)$, of the form
\[
([g],[g],[t]) \in \pi_k(M)\times \pi_k (M) \times \pi_k(N)
\]
where $[t]$ is just an element $\pi_k(N)$. Since $\phi$ belongs in the identity component of $\textup{Sympl}(M\times M \times N)$, composing with $\phi$ cannot change the homotopy class. Hence $[g]=0$, a contradiction.

This shows all homotopy groups of $M$ are trivial, by Whitehead's theorem this means $M$ is contractible.
\end{proof}

We next exhibit some familiar physical systems that cannot be cloned:
\begin{example}
\begin{enumerate}
    \item The simple pendulum has phase space $(T^*S^1,d\lambda)$, where $\lambda$ is the canonical 1-form.
    \item The sphere pendulum has phase space $(T^*S^2, d\lambda)$.
    \item The double pendulum has phase space $(T^*T^2,d\lambda)$.
\end{enumerate}
\end{example}

\begin{remark} \label{remark:comparison}
    In \cite{REDDY2019125846}, they provide cloning results satisfying Definition 4 in \cite{fenyes} for phase spaces of type $(T^*G,d\lambda)$, where $G$ is a Lie group. For instance $G = SU(2) \cong S^3$. Theorem \ref{thm:clone} shows a cloning in the sense of Definition \ref{def:clone} cannot exist. This in particular shows the symplectomorphism $\phi \in \textup{Sympl}(T^*S^3 \times T^*S^3 \times T^*S^3)$ constructed in \cite{REDDY2019125846} does not lie in the identity component of $\textup{Sympl}(T^*S^3 \times T^*S^3 \times T^*S^3)$, and in particular cannot come from the flow of a Hamiltonian.
\end{remark}

\section{Topological obstructions to approximate cloning}

Even though quantum cloning is not possible, it is of great interest to quantum information theory whether \emph{approximate quantum cloning} is possible. We formulate the analogous problem for classical cloning. Because our methods are inherently topological, we are able to provide quite strong obstructions to even approximate cloning when our phase space has nontrivial topology.

\begin{definition} \label{def:approximate cloning}
    Suppose we are given
    \begin{itemize}
    \item A symplectic manifold $(M,\omega_1)$, which we think of as the system being cloned;
    \item A symplectic manifodl $(N,\omega_2)$, which we think of as the cloning machine.
    \end{itemize}
    Let $h_1, h_2 \in \mathrm{Diffeo}_0(M)$, where $\mathrm{Diffeo}_0(M)$ denotes the identity component of the diffeomorphism group of $M$. We say $(M,\omega_1)$ can be approximately cloned with error $(h_1,h_2)$ using cloning machine $(N,\omega_2)$ if there are points $b\in M$, $r\in N$, and a symplectomorphism $\phi$ in the identity component of $\mathrm{Sympl}(M\times M \times N)$ that sends
    \[
    (x,b,r) \rightarrow (h_1(x),h_2(x),f(x,b,r))
    \]
    for all $x\in M$. Here $f(-,b,r)$ is a smooth function from $M$ to $N$.

    If there exists an $(N, \omega_2)$ and $h_1,h_2 \in \mathrm{Diffeo}_0(M)$ for which $(M,\omega_1)$ can be approximately cloned with error $(h_1,h_2)$, then we simply say $(M,\omega_1)$ can be approximately cloned.
\end{definition}
We note the above definition encompasses all $C^1$-small deformations of the original cloning map as defined in Definition \ref{def:clone}. This means cloning with small errors measured in the $C^1$ norm are included. However,
observe we have placed no size constraints on the distance between $(h_1,h_2)$ and the identity map - we only require they be connected to the identity. Hence our notion of cloning with error is fairly robust and includes maps that are far away from the identity.

\begin{theorem} \label{thm:approximately clone}
If $(M,\omega_1)$ can be approximately cloned, then $M$ is contractible.
\end{theorem}

\begin{proof}
    The proof is largely analogous to the proof of Theorem \ref{thm:clone}. Let $[g]: S^k \rightarrow M$ denote a nontrivial element of $\pi_k(M)$. Consider the map $(g,b,r):S^k\times \{pt\}\times \{pt\} \rightarrow M\times M\times N$ representing the homotopy class $([g],0,0) \in \pi_k(M\times M \times N)$. Then under the cloning map, this element of $\pi_k(M\times M \times N)$ is sent to
    \[
    \phi \circ (g,b,r): (x,b,r) \rightarrow (h_1\circ g(x), h_2 \circ g(x), r').
    \]
    Since $h_1$ and $h_2$ are in the identity component, this also represents an element of $\pi_k(M\times M \times N)$ of the form $([g],[g],[t])$. Since $\phi$ is in the identity component, this means $[g]=0$. By Whitehead's theorem, $M$ is contractible.
\end{proof}

\section{Physical Interpretation of our results} \label{physical interpretation of our results}
Our results about classical no cloning are counter-intuitive. We see copying/cloning of classical states all around us: a computer is able to copy data (which we think of a string of 1's and 0's) without changing the original data. It seems very conceivable one can build an analogue of a computer in the framework of classical mechanics; hence it is worthwhile to examine how to reconcile the existence of a computer with our no-cloning results.

\subsection{Cloning with a computer}

We provide a thought experiment that re-conciliates our no-cloning results with our classical intuition in certain cases. See the discussion posted on \cite{Baez_2023} \footnote{Most of the content of this section arose from a series of comments I posted on the n-Category Cafe in response to Jeffery Winkler and John Baez's discussion why in the real world we never worried about the classical no-cloning theorem. See \cite{Baez_2023}}.

We loosely think of the cloning process of a computer as someone inputting a string of 0s and 1s, and the computer copying the identical strings of 0's and 1's onto another system without destroying the original sequence. 

We build a mechanical representation of the above process as follows. We imagine there is a single coin, which we label coin A, sitting on a table (we can generalize to $N$ coins for a string of ``0''s and ``1''s, but we just say one for brevity). There is another identical coin labelled coin B on the same table. We think of coin A as the system we would wish to clone, and coin B as the system we wish to convert to the state of coin A. 

A coin being ``heads'' we think of a representing ``1'' and a coin lying ``tails'' as representing ``0''. The cloning process then involves building a cloning machine, and writing down a Hamiltonian so that at $t=0$ coin $B$ is heads, and coin A is either heads or tails. After time evolution, at time $t=1$, we require the state (of being either heads or tails) of coin B to agree with that of coin A.

This is always possible, because the state of being heads or tails of a coin is represented by \emph{discrete points} in the phase space of the coin. The cloning process as we described, which is meant to model the copying process of a computer, is mathematically realized by moving a collection of \emph{discrete points} to another collection of discrete points in phase space. It is a theorem in symplectic geometry that we can always write down a Hamiltonian that moves any finite collection of discrete points into another finite collection of points (of the same cardinality) regardless of how complicated the topology of the phase space is. See for example Theorem A in \cite{transitivegroup}. Our no-cloning theorems \ref{thm:clone} \ref{thm:approximately clone} are concerned with cloning entire regions of phase space that have nontrivial topology. 

The lesson gleaned from this example is that if a physical system is constrained to move only in a small (and topological trivial) region of its phase space, then cloning it in classical mechanics can still be possible (in the coin example the coin lying on a table either heads or tails occupy discrete points in its phase space); however if the system is allowed access the entire phase space and that phase space had nontrivial topology (for the example of a coin, if it were allowed to rotate and move freely, its phase space would be $T^*(\mathbb{R}^3 \rtimes SO(3))$), then our no-cloning theorem would say cloning is impossible.

The reason we would imagine cloning in the classical world is always possible is because when we perform cloning classically, we are almost always in the former case - we always prepare our systems to lie in discrete stationary points in phase space instead of moving unconstrained in phase space. In other words, we always wait for our system to ``settle down'' to some stationary state (for the example of the coin, heads or tails) before we attempt any kind of cloning procedure. These stationary discrete states are natural for us to consider essentially because the classical world is \emph{dissipative} - all systems will eventually settle into these stationary states if we wait long enough for the energy to dissipate.

\subsection{A quantum analogy} \label{subsec:a quantum analogy}
As observed above, if we restrict to small regions of phase space with trivial topology, we can still perform cloning. However if the system can access large regions of phase space with nontrivial topology, then cloning is not possible. We note here there is a simple quantum analogue. Given a Hilbert space $\mathcal{H}$, it is always possible to copy a state lying on a one dimensional subspace (which is analogous to the preparation step in any quantum experiment) despite the quantum no-cloning theorem preventing the cloning of the entire Hilbert space. We explain this in more detail as follows.

Let $\mathcal{V} \subset \mathcal{H}$ be a one dimensional complex subspace. We take our cloning machine to be $\mathbb{C}$, with the initial state of the cloning machine $\ket{1} \in \mathbb{C}$.

Then we claim there exists a vector $\ket{b} \in \mathcal{H}$, and a unitary operator $U$ on $\mathcal{H}\otimes \mathcal{H} \otimes \mathbb{C}$ so that for any $\ket{\psi} \in \mathcal{V}$, the operator $U$ takes

\[
\ket{\psi} \otimes \ket{b} \otimes \ket{1} \rightarrow \ket{\psi}\otimes \ket{\psi} \otimes \ket{r'}
\]
This is by definition the cloning of the one dimensional vector space $\mathcal{V}$.
Such an operation is easily constructed. Take $\ket{b} \in \mathcal{V}$, and $U$ to be the identity map. Since $\psi \in \mathcal{V}$, there exists $c\in \mathbb{C}$ so that $\ket{\psi} = c\ket{b}$. Then the identity map sends (assuming $c\neq 0$)
\[
\ket{\psi} \otimes \ket{b} \otimes \ket{1}  \rightarrow \ket{\psi} \otimes \frac{1}{c} \ket{b} \otimes c\ket{1} = \ket{\psi}\otimes \ket{\psi} \otimes c\ket{1}.
\]
The case $c=0$ corresponds to $\ket{\psi}=0$. This is just the fact the identity map sends 0 to 0. We note we could have in fact picked $\ket{b}$ to be any element of $\mathcal{H}$: we can first rotate $\ket{b}$ into an element of $\mathcal{V}$ using an unitary transformation, then apply the process above.

\section{future directions}\label{sec:future directions}

\subsection{Cloning in exotic geometries}
While there is a cloning result for the standard symplectic structure on $\mathbb{R}^{2N}$ \cite{fenyes}, it is known there are certain exotic symplectic structures on $\mathbb{R}^{2N}$ \cite{exotic_symplectic} which are not symplectomorphic to the standard symplectic structure. The construction is nontrivial because locally all symplectic manifolds look alike. As we have noted in this paper, the process of cloning also depends on the global topology of the symplectic manifold. It is hence natural to wonder if the cloning process ``sees'' that the global symplectic structure is exotic. To be precise:
let $(\mathbb{R}^{2N}, \omega)$ be an exotic symplectic structure on $\mathbb{R}^{2n}$, can it still be cloned?

\subsection{Free energy for cloning}

An area of great current interest is the thermodynamic aspect of information processing \cite{information_thermo,landauer,thermodynamics_of_feedback_controlled_systems}; see the review \cite{information_thermo} and the references therein. One basic principle is that information processing, such as measurement, does not happen for free: they come at a cost, to appropriately defined notions of free energy (\cite{information_thermo}).

We can think of cloning as a type of information processing. We start with the information of the original system, and copy it into another system without altering the original copy. The general expectation in thermodynamics is that this process would cost (free) energy (\cite{information_thermo}. The question is then what would be the correct expression for the free energy associated to the cloning process. 

We observe here that the usual method of measuring energy known to symplectic geometers via the Hofer energy \cite{hofer_1990} is not applicable in this setting. To be precise, suppose $U$ is a compact region in $\mathbb{R}^{2n}$. We use cloning machine $(N,\omega')$. Let $\textup{Ham}$ denote the space of all Hamiltonian diffeomorphisms. We say $\phi' \in \textup{Ham}(\mathbb{R}^{2n}\times \mathbb{R}^{2n} \times N)$ clones $U$ if there exists $(b,r) \in \mathbb{R}^{2n}\times N$
    so that $\phi'$ sends
    \[
    (x,b,r) \rightarrow (x,x,r')
    \]
    for all $x\in U$. 
Then by the Theorem 1.1 of \cite{Usher}, the infinum of the Hofer energy of all $\phi'$ that clone $U$ is zero. Hence in order to capture the idea of a ``free energy'' for the cloning process, a new definition of energy will be needed.
\subsection{No cloning and the classical measurement problem}

While we have explained in Section \ref{physical interpretation of our results} how to reconcile in some basic cases our no-cloning theorem with the intuition that the classical world is clonable, more ``paradoxes'' emerge if we probe further.

We propose the following thought experiment. Imagine walking into a laboratory with two identical spherical pendulums labelled A and B. Pendulum A is allowed to be in any initial state (in its phase space), and pendulum B is initially at rest in its ground state. We wish to clone the state of pendulum A with pendulum B. One would classically expect to be able ask an experimenter (or rather, to build a classical machine) to
\begin{procedure}[Cloning without errors]

\begin{enumerate}
    \item \textup{First measure precisely the position and momentum of pendulum A without perturbing the state of pendulum A;}
    \item \textup{Then push pendulum B into having the same position and momentum in its own phase space. Note this does not mean pendulum B has to occupy the same physical point as pendulum A (the two pendulum are assumed to be far apart from each other so they never collide). Since the two pendulums are identical, their phase spaces are identified with each other - we only need to push pendulum B into the corresponding point in its own phase space.}
\end{enumerate}
\end{procedure}
If we can simply have an experimenter follow the above procedure, or build a classical mechanical system to carry out the above procedure, then we would have performed a cloning, in violation of Theorem \ref{thm:clone}. In fact, within the world of idealized classical physics experiments, neither of the above two steps is particular unreasonable: we always assume that given any state of a classical system, we can measure it without changing the original state; further if given any fixed position and momentum, we can always push a particle into having the given position and momentum.

However, our main theorem \ref{thm:clone} tells us the cloning process, described as the above two procedure, simply cannot be done in the framework of Hamiltonian dynamics. Hence however simple the above two steps sound, we cannot build a classical machine to perform them in the way we have described-  at least one of the two steps above must fail in Hamiltonian dynamics. However, it is not at all apriori clear which of it will fail and in what form. Before delving into this further, we use the ``no-approximate-cloning theorem'' to deduce an even stronger paradox.

One could conceivably argue that if we try to build a machine to realize the above cloning procedure, the issue is that no finite dimensional machine can measure the state of a system (or push a pendulum into a state ) with infinite precision - that is simply an artifact of our idealization. Thence one would argue because since both steps above can only be performed up to finite errors, we don't have a contradiction to Theorem \ref{thm:clone}. However with the strength of Theorem \ref{thm:approximately clone} on no approximate cloning, we are able to move away from the idealized setup and accommodate the presence of errors. In particular, we can consider the following procedure:
\begin{procedure}[Cloning up to small error]
\begin{enumerate}
    \item \textup{First measure the position and momentum of pendulum A up to some small error; this measurement process can also perturb the state of pendulum A up to some small error.}
    \item \textup{ Next push pendulum B into the measured position and momentum of pendulm A, also up to some small error.}
\end{enumerate}
\end{procedure}
If we can build a classical machine to perform the above procedure, and in each step the size of the error can be accommodated within the framework of Definition \ref{def:approximate cloning} (for example if each error term is small in the $C^1$ norm, which seems very reasonable to achieve), then being able to perform the above procedure implies that one is able to perform \emph{approximate cloning}; however, this is prohibited by Theorem \ref{thm:approximately clone}.

Hence we see that either in the framework of precise cloning or approximate cloning, our no-cloning results clash sharply with our intuition about what kind of experimental operations can be performed in classical mechanics. To more clearly examine the issue, we make 
our questions more precise by formulating the above cloning process in a more mechanical fashion.
\begin{procedure}[A mechanical realization of cloning]
\begin{enumerate}
    \item \textup{In mechanics the way we imagine measuring the position and momentum of the rotating spherical pendulum is by scattering a large (yet finite) number of very light particles onto it. Then the information about the position and momentum of the original spherical pendulum is contained in the position and momentum of the scattered particles.}
    \item  \textup{We next imagine building a machine and a Hamiltonian to extract this information out of the scattered particles, and sending an identical pendulum into the same position and momentum. }
\end{enumerate}
\end{procedure}

We know somewhere along this process something must go wrong because of Theorems \ref{thm:clone} \ref{thm:approximately clone}. Aprior there are two obvious guesses as to what may be the issue:
\begin{enumerate}
\item During the process of scattering, because we are only allowed to scatter a finite number of particles, we are not able to obtain enough information about the original system through the position and momentum of the scattered particles. If this is the issue, the approximate no cloning theorem should imply that something must go seriously wrong in the measurement process (because we are allowed to have small errors in the cloning result) - i.e. two very distinct points in the phase space of the spherical pendulum results in nearly identical scattering data.
\item If we can obtain complete information (or just good enough information) about the spherical pendulum via scattering, then the issue might be that it is not possible to extract this data from the scattered particles and push another spherical pendulum into the same position and momentum by turning on a Hamiltonian.
\end{enumerate}

At the present it is not clear which of the two process fails; we know the failure of cloning has to do with the fact the phase space of the spherical pendulum has nontrivial topology, though it is hard to see the effect of that directly. The content of the first point may be studied through the lens of scattering theory see for example \cite{Novikov,canonical_scattering_trans_classical_mech,inverse_problems_class_scattering,qual_aspect_class_potential_scatter,S_matrix_in_classical_mech,the_scattering_matrix_and_formulas,time_dependent_scattering} for a treatment of scattering theory in classical mechanics. It is not to our knowledge that any work has been done in the direction where the scattering target has non trivial topology in their phase space. However, we expect that if the first point is the only obstruction to cloning, it should be the case that there are distinct regions in phase space that have near identical scattering outcomes; and this should be observable via numerical experiments.

It is not clear what kind of tools would be needed to help understand second bullet point, though we suspect it might be related to information theory \cite{information_thermo} applied to the classical mechanics setting: can we always efficiently retrieve data that is distributed into the position and momentum data of a collection of scattered particles?

It would seem a vital step in re-conciliating our classical intuition about cloning with the theorems proved in this paper would be a better understanding of the classical measurement process. We leave this direction for future research.
\begin{acknowledgments}
I would like to thank John Baez whose comments helped greatly improve this paper. I also thank Tomohiro Soejima for helpful comments. I would also like to thank the anonymous referee whose comments helped improve the content and the exposition of this paper. I would also like to thank Vincent Humili\`ere for telling me about Michael Usher's result.
\end{acknowledgments}

\nocite{}
\nocite{*}
\bibliography{aipsamp}

\end{document}